\documentclass[12pt]{article}
\usepackage[margin=1in]{geometry}
\usepackage{natbib}
\usepackage{macros}
\usepackage{amsmath, amssymb, graphicx, amsthm}
\usepackage{tikz-cd}
\tikzset{>=stealth}
\usepackage{pgfplots}
\usepackage{float}
\usepackage{hyperref}     
\usepackage{url}            
\usepackage{booktabs}       
\usepackage{microtype}      
\usepackage{xcolor}
\usepackage{caption}
\usepackage{subcaption}
\usepackage{wrapfig}

\newtheorem{thm}{Theorem}
\newtheorem{prop}{Proposition}

\theoremstyle{definition}
\newtheorem{assump}{Assumption}

\theoremstyle{remark}
\newtheorem*{remark}{Remark}

\DeclareMathOperator{\argmin}{argmin}
\DeclareMathOperator{\argmax}{argmax}

\title{Distance-to-Set Priors and Constrained Bayesian Inference}
\author{Rick Presman\\ Department of Statistical Science, Duke University \and Jason Xu \\ Department of Statistical Science, Duke University}
\date{\today}

\begin{document}
\maketitle

\begin{abstract}
    Constrained learning is prevalent in many statistical tasks. Recent work proposes distance-to-set penalties to derive estimators under general constraints that can be specified as sets, but focuses on obtaining point estimates that do not come with corresponding measures of uncertainty. To remedy this, we approach distance-to-set regularization from a Bayesian lens. We consider a class of smooth \textit{distance-to-set priors}, showing that they yield well-defined posteriors toward quantifying uncertainty for constrained learning problems. We discuss relationships and advantages over prior work on Bayesian constraint relaxation. Moreover, we prove that our approach is optimal in an information geometric-sense for finite penalty parameters $\rho$, and enjoys favorable statistical properties when $\rho\to\infty$. The method is designed to perform effectively within gradient-based MCMC samplers, as illustrated on a suite of simulated and real data applications.
\end{abstract}
\newpage
\section{Introduction}

Constrained learning is ubiquitous in statistical tasks when seeking to impose desired structure on solutions. Concretely, consider the task of estimating a parameter $\bx\in\bbR^d$ by minimizing some loss function $f(\bx)$ where $\bx$ needs to satisfy a set of constraints encoded by a set $\calC$. Then we seek:
\begin{equation}\label{eq:constr}
    \min_{\bx}\;  f(\bx)  \quad  \st \; \bx\in\calC
\end{equation}
A simple but powerful observation that will make this amenable to effective algorithms is to equivalently express the restriction in terms of the Euclidean distance between the point $\bx$ and the constraint set as $\dist(\bx,\calC) = 0$.
In many instances, it is enough to only approximately satisfy the constraints. A recent framework that accomplishes this kind of constraint relaxation is known as the distance-to-set penalization \citep{chi2014distance,xu2017generalized}: for $\rho\in(0,\infty)$,
$$
\bx^*\in\argmin_\bx\left[f(x) + \frac{\rho}{2}\dist(\bx,\calC)^2 \right].
$$
Solutions to this problem can be obtained using a majorization-minimization (MM) scheme  known as the proximal distance algorithm \citep{keys2019proximal}, and it is so called because the iterative updates are defined via proximal operators \citep{parikh2014proximal}. However, despite its ability to deliver point estimates effectively, it is very difficult to derive measures of uncertainty, and so a general theory of inference is difficult to obtain. Toward filling this methodological gap, we recast the optimization problem in a constrained Bayesian setting by an analog of these penalties that we term \textit{distance-to-set priors}.\\

Our approach draws previously unexplored connections between this optimization framework and the broader constrained Bayesian inference literature \citep{ghosh1992constrained,gramacy2016modeling}, an area that continues to grow with exciting recent ideas. We focus on a tradition of sampling through gradient-based samplers such as Hamiltonian Monte Carlo, or HMC \citep{neal2011mcmc, betancourt2015hamiltonian}. 
\citet{lan2014spherical} utilize a spherical HMC, mapping constraints that can be written as norms onto the hypersphere. Related recent work uses a Riemannian HMC under a manifold setup \citep{kook2022sampling}, extending a line of work pioneered by \citet{byrne2013geodesic}. \cite{duan2020bayesian} replace a support constraint with a term that decays exponentially outside of the support, while \cite{sen2018constrained} project sample draws from unconstrained posteriors to the constraint set to approximate the original posterior.  Recently, \cite{xu2021bayesian} propose using priors based on proximal mappings related to the constraint sets.
Concurrent work in \citet{zhou2022proximal} propose using a the Moreau-Yosida envelope more generally, using a class of epigraph priors toward regularized and constrained problems suited for proximal MCMC \citep{pereyra2016proximal}.\\

Distance-to-set priors extend this line of inquiry, providing an effective, practical way to consider constrained inference problems. The framework is more general than many of the previous methods in that it essentially only requires that the constraint can be written as a set, and that projection onto that set is feasible. These priors are then easy to evaluate, and work well within gradient-based samplers due to our smooth formulation. This improves computational stability under posterior sampling algorithms such as HMC, as we investigate in an empirical study.
Moreover from a theoretical perspective, this class of priors admits posteriors that converge in distribution to the original constrained problem along with their maximum a posteriori (MAP) estimates as we increase the parameter governing the degree of constraint enforcement. Finally, we draw a connection between Bayesian constraint relaxation and information geometry, revealing how distance-to-set priors are optimal in a certain sense, while simultaneously yielding a way to select the regularization parameter $\rho$ systematically.

\section{Theory and Methods}

We begin by briefly reviewing distance-to-set penalties and some of their key properties.

\paragraph{Distance-to-Set Penalties}

Let $\calC\subset\bbR^n$ be convex, and let $f:\calC\to \bbR$ be a convex function. Many constrained programming problems of the form \eqref{eq:constr}
may be intractable in their original form, but can be converted to a sequence of simpler subproblems. To make progress,  denote the Euclidean distance from any $\bx\in\bbR^n$ to $\calC$ by $\dist(\bx,\calC) := \inf_{\by} \|\bx-\by\|_2$: then the condition $\bx\in \calC$ can be equivalently written as $\dist(x,\calC) = 0$. Note that while the distance operator is not necessarily smooth, its square is differentiable as long as the projection of $\bx$  onto $\calC$, denoted $P_{\calC}(\bx) := \arg\min_{\by\in\calC} \|\bx-\by\|_2$, is single-valued (\cite{lange2016mm}).
Thus, we may reformulate the problem by instead considering the \textit{smooth} unconstrained optimization task:
\begin{align*}
\min_{\bx}\; & \left[f(\bx) + \frac{\rho}{2}\dist(\bx,\calC)^2\right],
\end{align*}
where $\rho>0$ is a penalty parameter.
To solve the resulting problem,  \cite{lange2016mm} propose a method termed the proximal distance algorithm which makes use of the MM principle to create surrogate functions based on distance majorization \citep{chi2014distance}.
Its namesake derives from the fact that the minimization of the surrogate functions
$$g_\rho(\bx \mid \bx_k )= f(\bx) + \frac{\rho}{2} \lVert \bx - P_\mathcal{C}(\bx_k) \rVert^2$$
is related to the proximal operator of $f$ (\cite{parikh2014proximal}): recall for a function $f$, the proximal mapping with parameter $\lambda$ is defined
$$\text{prox}_{\lambda f}(\by) \equiv \underset{\bx}\argmin \,\Big[ f(\bx) + \frac{1}{2\lambda} \lVert \bx - \by \rVert_2^2\Big],$$ which relates to our problem with $\by$ the projection at iterate $k$ and 
$\lambda=\rho^{-1}$. 
Under this formulation, to recover the solution to the original optimization problem, it is necessary for $\rho\to\infty$ at some appropriate rate (\citet{wright1999numerical}). Conversely, fixing a finite $\rho$ results in a solution where $x$ is close to $\calC$, but not strictly inside of the set. Both cases may be of interest depending on the modeling context. We will discuss primarily the latter in this paper but also establish theoretical relationships to the former.  As our primary setting is statistical, we may think of $f(\bx)$ as a convex loss function.

\subsection{Distance-to-Set Priors}

The proximal distance algorithm mentioned above provides a method for obtaining point estimates under distance-to-set penalization. However, to the best of our knowledge, the current literature does not provide results pertaining to uncertainty quantification for these estimators. Toward understanding their uncertainty properties, our first contribution is to link these ideas to a Bayesian constraint relaxation framework. Identifying a penalized estimation problem with a Bayesian problem has been done at least as early as the seminal LASSO paper \citep{tibshirani1996regression}. Consider data $\by\mid\btheta\in\bbR^n$ that has likelihood $L(\btheta\mid\by)$ and is parameterized by some parameter $\btheta$ with prior $\pi(\btheta)$ that is absolutely continuous with respect to Lebesgue measure, with support $\bbR^d$ but constrained to $\bTheta\subset\bbR^d$. Since $\btheta$ is constrained to $\bTheta$, Bayes' Theorem gives the posterior for $\btheta$:
\begin{align*}
\overline{\pi}(\btheta\mid \by) & \propto L(\btheta\mid\by)\pi(\btheta)\mathbf{1}_{\btheta\in\bTheta} \\
& \propto L(\btheta\mid\by)\pi(\btheta)\mathbf{1}_{\dist(\btheta,\bTheta)=0}
\end{align*}
where the second line follows from the discussion of distance-to-set penalties. Since sampling from a posterior sharply constrained on $\bTheta$ may be difficult, we can replace the indicator representing the constraint with $\exp\left(-\frac{\rho}{2}\dist(\btheta,\bTheta)^2 \right)$. An illustration is given in Figure \ref{fig:prior_graph}. This term is equal to the indicator on $\bTheta$ and rapidly decays to zero as the distance from $\btheta$ to $\bTheta$ grows larger. We choose to square the distance-to-set operator to align with the distance-to-set optimization and to improve sampling performance, which we discuss in Section \ref{comp_imp}.\\

\begin{figure}[h]
\begin{center} \vspace{-5pt}
    \begin{tikzpicture}[scale=.98]
    	\begin{axis}[xmin=-3, xmax=3, ymin=-0.5, ymax=1.5, samples=50, axis lines=middle]
    		\addplot[red, very thick, domain=-3:-1]{e^(-(x+1)^2/2)};
    		\addplot[red, very thick, domain=1:3]{e^(-(x-1)^2/2)};
    		\addplot[red, very thick, domain=-1:1](x,1);
    		\addplot[blue, very thick, dashed, domain=-3:-1](x,0);
    		\addplot[blue, very thick, dashed, domain=0:1](-1,x);
    		\addplot[blue, very thick, dashed, domain=-1:1](x,1);
    		\addplot[blue, very thick, dashed, domain=0:1](1,x);
    		\addplot[blue, very thick, dashed, domain=1:3](x,0); 
    		\end{axis}
    \end{tikzpicture} \vspace{-10pt}
\end{center}
\caption{Schematic of distance-to-set prior for the set $\bTheta = [-1,1]$. The blue dashed line represents the indicator $\1_{\bTheta}$, and the red solid line represents the relaxation of $\1_{\bTheta}$ that we consider in this paper. } 
\label{fig:prior_graph}
\end{figure}
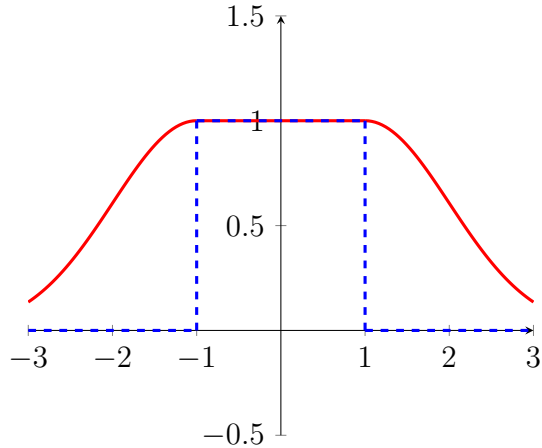

We propose to use \textit{distance-to-set priors}, defined as follows:
$$
\widetilde{\pi}(\btheta) := \pi(\btheta)\exp\left(-\frac{\rho}{2}\dist(\btheta,\bTheta)^2 \right),
$$
where $\rho>0$ is a hyperparameter in our treatment. To bridge this with the optimization setting, we can define $f(\btheta) := - \log (L(\btheta\mid\by)\pi(\btheta))$. However, it should be noted that although every likelihood function gives us a loss by taking the negative-log, the converse is not always true. Thus, one could generalize our approach by considering more general loss functions and incorporating them into the Bayesian framework via Gibbs posteriors \citep{bissiri2016general}; see also \cite{jacob2017better}.
Ignoring the constraint for a moment, we obtained the \textit{unconstrained posterior}:
\begin{equation}\label{eq:unconstrained}
\pi(\btheta\mid\by) \propto L(\btheta\mid\by)\pi(\btheta).
\end{equation}
Combining this with constraint relaxation, we obtained what we call the \textit{constraint relaxed posterior}:
\begin{equation}\label{eq:relaxed}
\widetilde{\pi}(\btheta\mid\by) \propto L(\btheta\mid\by)\pi(\btheta) \exp\left(-\frac{\rho}{2}\dist(\btheta,\bTheta)^2 \right)
\end{equation}

The relaxed form avoids the discontinuity implied by the indicator function in the sharply \textit{constrained posterior} :
\begin{equation}\label{eq:constrained}
\overline{\pi}(\btheta\mid\by) \propto L(\btheta\mid\by)\pi(\btheta)\1_{\btheta\in\bTheta}
\end{equation}
For the remainder of this paper, we make the following assumptions:

\begin{assump}\label{assump1}
    All probability measures are absolutely continuous with respect to $d$-dimensional Lebesgue measure with densities supported in $\bbR^d$.
\end{assump}

\begin{assump}\label{assump2}
    The unconstrained posterior $\pi(\btheta\mid\by)$ is proper; that is, $\displaystyle
    \int_{\bbR^d} L(\btheta\mid\by)\pi(\btheta)\,d\btheta < \infty.
    $
\end{assump}

\begin{assump}\label{assump3}
    Unless stated otherwise, the support $\varnothing\neq\bTheta\subset\bbR^d$ is a closed and convex set.
\end{assump}

We make Assumption \ref{assump1} because we will have an interest in the performance of samplers, like HMC, that are designed to perform on continuous distributions and to simplify the setting. Assumption \ref{assump2} guarantees the original posterior is not ill-posed, and ensures we are sampling from a well-defined distribution.
Assumption \ref{assump3} plays an integral role toward the smoothness properties of the constraint relaxation; they lead to continuity of the projection as well as a unique gradient of the squared distance.\\

These natural conditions asure that the object of interest is well-defined. The following proposition, as well as all theorems in the following section, are proven in the Appendix.

\begin{prop}\label{proper}
Under Assumptions \ref{assump1} and \ref{assump2}, the constraint relaxed posterior $\widetilde{\pi}(\btheta\mid\by)$ (Equation \ref{eq:relaxed}) is a proper density.
\end{prop}

\subsection{Statistical Properties}

Distance-to-set regularization and the underlying constrained problem are inextricably link, so one would naturally hope that the constraint relaxed posterior behaves approximately like the constrained posterior when $\rho$ is large. Fortunately, this is true as we formalize in the guarantees below. Our first result shows that our class of distance-to-set priors also posses the desirable property that the sequence of MAP estimators of the relaxed posterior (indexed by the penalty parameter $\rho$) converge to the the MAP estimator of the non-relaxed problem as $\rho$ grows large when the posterior is log-concave.

\begin{thm}
Suppose the unconstrained posterior $\pi(\btheta\mid \by)$ (Equation \ref{eq:unconstrained}) is strictly log-concave. Let $\{\widetilde{\pi}_{\rho_k}(\btheta\mid\by)\}_{k\in\bbN}$ (Equation \ref{eq:relaxed}) be a sequence of constraint-relaxed posterior distributions where $\rho_k\uparrow\infty$ as $k\to\infty$. Further, define the following MAP estimators
\[ \widehat{\btheta}^M = \argmax_{\btheta} \overline{\pi}(\btheta\mid\by),
   \quad  \widehat{\btheta}_{\rho_k}^M = \argmax_{\btheta} \widetilde{\pi}_{\rho_k}(\btheta\mid\by).\]
Then the sequence $\widehat{\btheta}_{\rho_k}^M \to \widehat{\btheta}^M$ as $k\to\infty$.
\end{thm}

In addition to convergence of a point estimate, we can say more about the behavior of the entire distribution. 

\begin{thm}
Let $\overline{\Pi}$ be the constrained posterior distribution with density $\overline{\pi}(\btheta\mid\by)$, and let $\{\widetilde{\Pi}_{\rho_k}\}_{k\in\bbN}$ be a sequence of constraint-relaxed posterior distributions with densities $\{\widetilde{\pi}_{\rho_k}(\btheta\mid\by)\}_{k\in\bbN}$, respectively, where $\rho_k\uparrow\infty$ as $k\to\infty$. Then  $\|\widetilde{\Pi}_{\rho_k} - \overline{\Pi}\|_{\mathsf{TV}} \to 0$ as $k\to \infty$. It follows that $\widetilde{\Pi}_{\rho_k}\overset{D}{\rightarrow}\overline{\Pi}$ as $k\to 0$.
\end{thm}

Theorem 2 is consistent with concurrent work by \cite{zhou2022proximal} showing convergence in total variation distance for posterior distributions under general epigraph priors.

\paragraph{Information Projection}
The preceding results primarily concern the limiting setting where $\rho$ grows large, confirming that the relaxed posteriors under our priors tend to the sharply constrained posterior. However, a common modeling application in practice entails selecting a finite value $\rho<\infty$ to promote structure encoded in the constraint $\mathcal{C}$. The next contribution highlights a deeper connection between constrained and constraint-relaxed posterior distributions from an information geometric perspective.\\

Consider the special case of the moment-constrained information projection problem, originally studied by \cite{csiszar1975divergence}:
\begin{align}\label{eq:inf}
    \min_{p(\btheta)}\;& \int p(\btheta)\log\left(\frac{p(\btheta)}{\pi(\btheta\mid\by)}\right)\,d\btheta \\
    \st\;& \bbE_{\btheta\sim p}\left[\frac{1}{2}\dist(\btheta,\bTheta)^2 \right] = D\nonumber
\end{align}
Thus we are interested in finding the closest density $p(\btheta)$ to the unconstrained posterior $\pi(\btheta\mid\by)$ in terms of KL divergence such that the expected square distance of $\btheta$ to $\bTheta$ under $p(\btheta)$ is equal to some given value $D$.

\begin{thm}\label{thm:lagrange}
Suppose that $\bbE_{\btheta\sim\pi(\btheta\mid\by)}[\dist(\btheta,\bTheta)^2/2] > D$. Then the constraint-relaxed posterior distribution $\widetilde{\pi}(\btheta\mid\by)$ (Equation \ref{eq:relaxed}) is the solution to the moment-constrained information projection problem (Equation \ref{eq:inf}):
$$
p^*(\btheta) \propto \pi(\btheta\mid\by)\exp\left(-\frac{\lambda}{2}\dist(\btheta,\bTheta)^2 \right),
$$
where $\lambda>0$ is a Lagrange multiplier that satisfies the moment constraint under $p^*(\btheta)$.
\end{thm}

The solution given in Theorem  \ref{thm:lagrange} is known as \textit{exponential tilting} \citep{west2020perspectives, tallman2022entropic}. Observe that for $\lambda=0$, $p^*(\btheta) = \pi(\btheta\mid\by)$. Moreover, $D$ and $\lambda$ are inversely related (see Appendix for additional details), so in particular, if $D\to0$, then $\lambda\to\infty$ and $p^*(\btheta) \to \overline{\pi}(\btheta\mid\by)$. Exponential tilting therefore creates a spectrum of constraint relaxation with the unconstrained posterior on one end, the constrained posterior on the other end, and the constraint-relaxed posterior as the optimal choice in the sense that it is the closest to $\pi(\btheta\mid\by)$ while maintaining a specified distance from $\bTheta$ in expectation.\\

This perspective has practical implications. A common challenge in regularization problems involves specifying the penalty parameter when it does not have an interpretable scale. Theorem \ref{thm:lagrange} provides a systematic solution by identifying the Lagrange multiplier $\lambda$ from the information projection with the penalty $\rho$. We can solve for $\lambda$ given a value for $D$, and then use that value as the corresponding value for $\rho$ in the distance-to-set regularization or the corresponding Bayesian constraint relaxation. Thought it appears we've simply swapped specifying $\rho$ with $D$, it's important to note that $D$ is often interpretable in practice as it is on the same scale as $\btheta$ interpretable scale, so we can choose the level of relaxation using real-world inputs in application.

\subsection{Prior work on Bayesian Constraint Relaxation}

The task our contributions address is closely related to the Bayesian constraint relaxation work by \citet{duan2020bayesian}. There, the authors also consider relaxing a sharply constrained prior by quantifying the distance to the desired constraint, with particular attention to the case when the constraint sets which they denote $D$ lie in a lower dimensional subspace of the full space $\mathbb{R}^d$. 
They construct posteriors of the form $\widetilde\pi_\lambda \propto \ell(\btheta; Y) \pi_R(\btheta) \text{exp} \{ -\lambda^{-1} \lVert \nu_D(\btheta) \rVert\}$, where $s<d$ denotes the dimension of the constraint set $D$, which is represented algebraically as a solution to the system of equations $\{ \nu_j(\btheta) = 0\}_{j=1}^s$. \citet{duan2020bayesian} choose to measure the constraint violation explicitly using the function $\lVert \nu_D(\btheta) \rVert = \sum_{j=1}^s | \nu_j(\btheta)|$.\\

The authors briefly comment that users may flexibly choose a measure of constraint violation: along this line, our method not only shows how the squared Euclidean distance is preferable in many ways over their choice of $\lVert \nu_D(\btheta) \rVert $, but makes a key departure from defining constraints algebraically and component-wise by grounding  in a \textit{projection-based} framework. That is, even when a constraint set $D$ has measure zero in  $\mathbb{R}^d$,  for any point $\bx \in \mathbb{R}^d$, its projection  $P_D(\bx) \in \mathbb{R}^d$ also lives in the  ambient space. By exploiting the projection-based characterization of the distance from points to sets, our formulation handles constraints \textit{implicitly}, yielding effective algorithms that stay in the original space. Not only does this avoid having to explicitly write constraints algebraically, but obviates technical geometric measure theoretic arguments by avoiding the need to operate directly in the subspace containing $D$ and resolve the mismatch in dimension when mapping back into $\mathbb{R}^d$.\\

Our work shares a connection with recent work that proposes a class of nondifferentiable priors called epigraph priors in the context of Bayesian trend filtering \cite{heng2022bayesian}. Though connections to proximal distance algorithms and distance majorization are not explicitly referenced by the authors, projection onto the epigraph of a regularization function $g$ depends on the proximal mapping of $g$, and the success of their framework hinges on the same algorithmic primitives and known projection operators or proximal maps that make computation attractive in our case. Indeed, the proximal map of an indicator function $1_{C}(x)$ of a set $C$ is given by the projection $P_C(x)$ onto $C$. From another perspective, the Moreau-Yosida envelope of $1_{C}(x)$ is given by the squared distance between $x$ and $C$. While neither of these discusses the Bayesian constraint framework of \cite{duan2020bayesian}, in concurrent work \cite{zhou2022proximal} also remark on the connection between distance to epigraph approaches and distance regularization from the optimization perspective.

\subsection{Sampling via Hamiltonian Monte Carlo}\label{comp_imp}

Having established its properties, we now discuss how to effectively draw samples from the posterior distribution in practice. We advocate Hamiltonian Monte Carlo (HMC) \citep{neal2011mcmc, betancourt2015hamiltonian}, a popular gradient-based MCMC algorithm that leverages Hamiltonian dynamics to generate effective parameter proposals.\\

We briefly review the HMC framework: to sample from a posterior $\pi(\btheta\mid\by) \propto L(\btheta\mid\by)\pi(\btheta)$, where the posterior has support on $\bbR^d$,
HMC begins by embedding $\btheta$ into $\bbR^{2d}$ via the introduction of an independent, auxiliary \textit{momentum} parameter $\bp\in\bbR^d$. The parameter of interest $\btheta$ plays the role of the \textit{position} vector; the sampler then explores their joint posterior: $\pi(\btheta,\bp\mid\by)$. Define the Hamiltonian function $H:\bbR^{2d}\to\bbR$ by $H(\btheta,\bp) := -\log\pi(\btheta,\bp)$. By the independence of $\btheta$ and $\bp$, we can write
$$
H(\btheta,\bp) = K(\bp) + U(\btheta),
$$
where one can take the kinetic energy to take the form $K(\bp) := \frac{1}{2}\bp^\intercal\bM^{-1}\bp + C$ for some constant $C$ and mass matrix $\bM$, and the potential energy $U(\btheta) := -\log\pi(\btheta\mid\by)$. The Hamiltonian dynamics that describe how the parameters evolve over ``time'' impose structure on the manifold containing $(\btheta,\bp)$:
$$
\begin{cases}
\frac{d\btheta}{dt} =  \nabla_{\btheta}H(\btheta,\bp)  = \nabla_{\btheta}\log\pi(\btheta\mid\by)\\
\frac{d\bp}{dt} = -\nabla_{\bp}H(\btheta,\bp) = -\bM^{-1}\bp
\end{cases}
$$
Generally, there is no analytical tractable solution for this PDE, so we rely on what is known as the leap-frog integrator to discretize the PDE as follows. Given some step size $\varepsilon$ and a number of steps $L$, we iterate  for $l=1,\ldots,L$:
\begin{enumerate}
    \item $\bp_{t+\varepsilon/2} = \bp_t - \left.\frac{\varepsilon}{2}\nabla_{\btheta}\log\pi(\btheta\mid\by)\right|_{\btheta = \btheta_t}$
    \item $\btheta_{t+\epsilon} = \btheta_t + \varepsilon \bM^{-1}\bp_{t+\varepsilon/2}$
    \item $\bp_{t+\varepsilon} = \bp_{t+\varepsilon/2} - \left.\frac{\varepsilon}{2}\nabla_{\btheta}\log\pi(\btheta\mid\by)\right|_{\btheta = \btheta_{t+\varepsilon}}$
\end{enumerate}
To incorporate this into a sampling algorithm, suppose we start with a current parameter draw $\btheta^{(s)}$. Draw $\bp^0\sim N_d(\0,\bM)$. Perform the leap-frog integrator to obtain a proposal $(\btheta^{(s+1)},\bp^*)$. After reversing the direction of momentum $-\bp^* \mapsto \bp^*$, we perform an accept-reject step to correct discretization error: accept $(\btheta^*,\bp^*)$ with probability
$$
\alpha = \min\left\{1, \frac{e^{-H(\btheta^{(s+1)},\bp^*)}}{e^{-H(\btheta^{(s)},\bp^0)}} \right\}.
$$

\paragraph{Computational Advantages} 

\cite{duan2020bayesian} report instability in the HMC algorithm, particularly when constraints are tightly enforced (i.e., $\rho$ is large) under their Bayesian constraint relaxation formulation. This section provides a simple explanation for this behavior by examining the gradients under each approach, and also reveals how our formulation avoids these by yielding continuously differentiable gradients. In doing so, we greatly improve stability in HMC implementations so that adequate mixing is not restricted to narrow parameter ranges.

\begin{prop}\label{smooth}
The log constraint-relaxed posterior $\log \widetilde{\pi}(\btheta\mid \by)$ (Equation \ref{eq:relaxed}) is continuously differentiable as long as the  log-posterior $\log\pi(\btheta\mid\by)$ (Equation \ref{eq:unconstrained}) is continuously differentiable in $\btheta$.
\end{prop}

The proof is detailed in the Appendix, but follows from continuity and uniqueness of the projection, which are given by convexity. In particular, we see that the gradient $$
\nabla_{\btheta}\left[\frac{1}{2}\dist(\btheta,\bTheta)^2 \right] = \btheta - P_{\bTheta}(\btheta)
$$
converges continuously to $0$ on the boundary of the constraint as desired.\\

To better understand  advantages over prior work, we examine how the gradient would behave had we relaxed the constraint without squaring a distance-to-set penalty, akin to an $\ell_1$ approach as in \citep{duan2020bayesian}. The log-posterior, denoted by $\widehat{\pi}(\theta)$ would be of the form:
$$
\log \widehat{\pi}(\btheta\mid \by) = \log L(\by\mid\btheta)\pi(\btheta) - \frac{\rho}{2}\dist(\btheta,\bTheta),
$$
which is not smooth in general. In particular, examining the subdifferential with respect to $\btheta$ yields
$$
\partial_{\btheta} \log\widehat{\pi}(\btheta) = \partial_{\btheta} \log L(\btheta\mid\by)\pi(\btheta) - \begin{cases}
\frac{\btheta - P_{\bTheta}(\btheta)}{\|\btheta - P_{\bTheta}(\btheta)\|_2}, & \btheta\not\in\bTheta \\
0, & \btheta\in\bTheta
\end{cases}
$$

Observe that the $\|\nabla_{\btheta} \dist(\btheta,\bTheta)\|_2 = 1$ for $\btheta\notin\bTheta$, and 0 otherwise: the distance fails to be continuously differentiable at the boundary, instead \textit{sharply transitioning} at a jump discontinuity. Computationally, this manifests as instability and poor mixing when the sampler is close to the constraint, as whenever $\btheta\approx P_{\bTheta}(\btheta)$, the denominator becomes numerically close to 0. This agrees with empirical findings reported in \citep{duan2020bayesian} and their remarks on instability in the Supplemental Materials.\\

\begin{remark}
We may weaken Assumption 3 so that $\bTheta$ is closed but not necessarily convex. In this case, Proposition 7 of \cite{keys2019proximal} assures that for a nonempty closed subset of $\bbR^n$, the projection operator is multi-valued on a set of measure zero, so the gradient formula for the squared distance function holds and is uniquely defined almost surely.
\end{remark}

\section{Results and Performance}

In this section, we investigate the performance of distance-to-set priors on increasingly more involved empirical studies. We find that our priors result in improved sampling performance relative to existing constraint relaxation methods.

\paragraph{Regression over the $\ell_2$-Ball}

We illustrate our approach to measuring the uncertainty of distance-to-set penalization by considering a simple constrained formulation of the ridge regression problem. Here the constraint set $\calC =  B_2(0,1)$ is the Euclidean unit $\ell_2$-ball. From an estimation perspective, the proximal distance algorithm aims to solve the problem
$$
\min_{\bbeta\in\calC} \;\|\by - \bX\bbeta\|_2^2,
$$
where $\by\in \bbR^n$, $\bX\in\bbR^{n\times p}$, and $\bbeta\in\bbR^p$, by  considering a relaxed version. For a fixed $\rho\in(0,\infty)$,
$$
\min_{\bbeta\in\bbR^p}\;\|\by - \bX\bbeta\|_2^2 + \frac{\rho}{2}\dist(\bbeta,\calC)^2
$$
Applying the iterations from the proximal distance algorithm without taking $\rho\uparrow\infty$ will solve this problem, the solution of which we denote by $\widehat{\bbeta}$. To obtain corresponding uncertainty as measured by a posterior, consider the Gaussian model: $\by\mid\bbeta \sim N(\bX\bbeta, \sigma^2\bI)$ with a flat prior $\pi(\bbeta)\propto 1$. Then the constraint relaxed posterior is given by

$$
\widetilde{\pi}(\bbeta\mid \by) \propto \exp\left(-\frac{1}{2\sigma^2}\|\by - \bX\bbeta\|_2^2 \right)\exp\left(-\frac{\rho}{2}\dist(\bbeta,\calC)^2 \right)
$$

Clearly, the MAP estimator $\widehat{\bbeta}_{\text{MAP}}$ is equal to $\widehat{\bbeta}$. Since we have a fully-specified posterior, we can supplement the estimator $\widehat{\bbeta}$ with uncertainty quantification. Moreover, a more subtle point is that we can use the proximal distance algorithm to compute MAP estimates for the corresponding Bayesian model, which would normally be very difficult to obtain simply from drawing samples from the posterior.\\

We examine the performance in this model on simulated data. In this case, there is a simple closed-form expression for the projection:
$$
P_{\calC}(\bbeta) = \begin{cases} \bbeta/\|\bbeta\|_2,& \bbeta\not\in\calC \\
0,& \bbeta\in\calC\end{cases}.
$$

\begin{figure}
    \centering
    \includegraphics[width=0.5\textwidth]{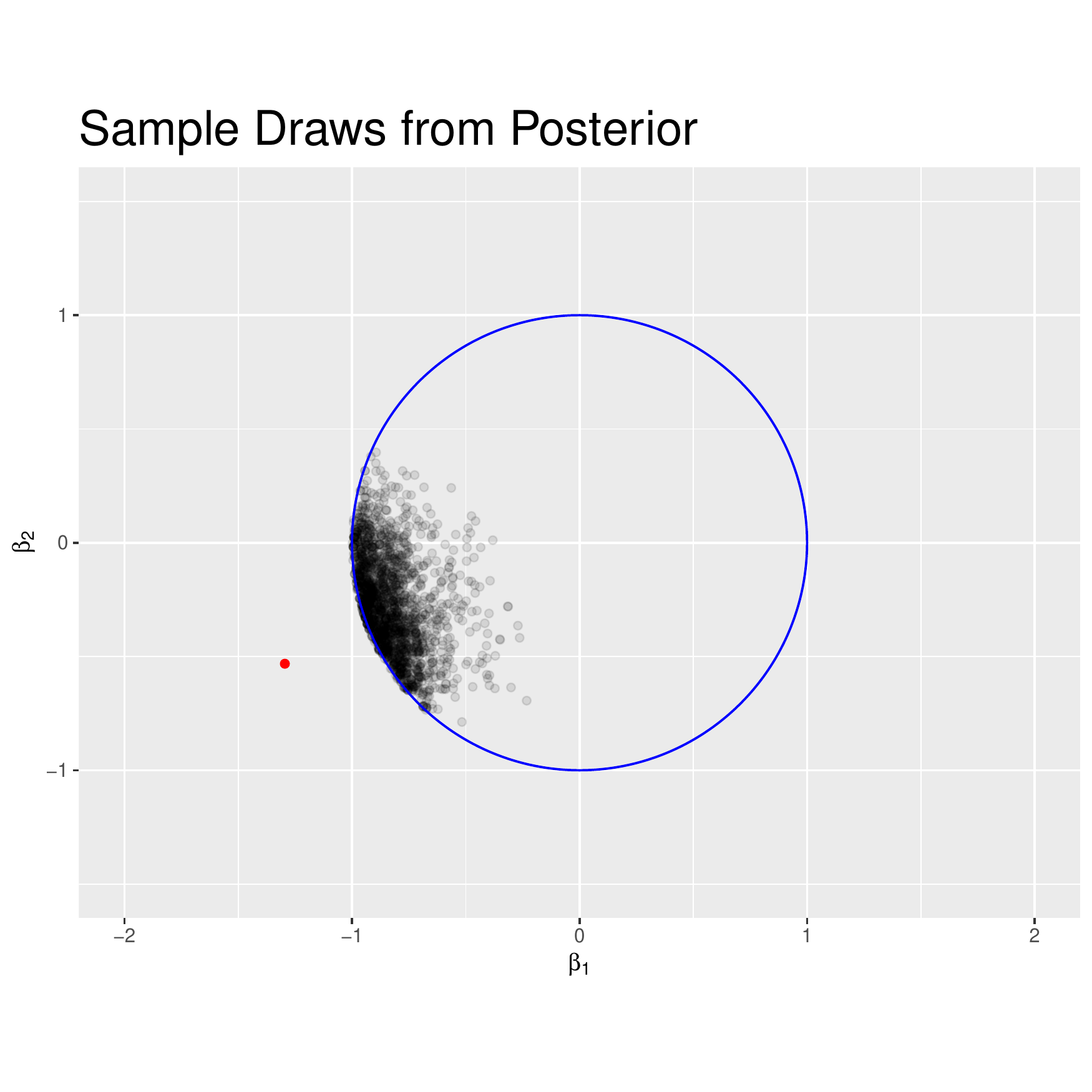}
    \caption{Draws from relaxed posterior, ridge regression.}
    \label{fig:opt}
\end{figure}

We choose $p=2$ for easy of visualization, and generate the true $\bbeta = (-1.295, -0.532)$ to lie outside of $\calC$. We then draw $n=100$ observations from a linear model under $\bbeta$. To sample from the posterior distribution, we use the \texttt{stan} functionality in \texttt{R} that leverages NUTS-HMC \citep{hoffman2014no}, and set the hyperparameter $\rho = 10^3$ to tightly enforce the constraint.

Figure \ref{fig:opt} displays the sample draws. At a glance, one can see that the posterior posterior distribution is concentrated near the boundary in the bottom-left quadrant since the true $\bbeta$, denoted as a red point, lies in that direction. The posterior samples allow one to conduct inference. For instance, the 95\% equi-tailed credible intervals for $\beta_1$ and $\beta_2$ are $(-0.99, -0.54)$ and $(-0.65, 0.11)$, respectively.\\

We further examine the impact of squaring the distance-to-set operator on posterior sampling performance in this context. Figure \ref{fig:acf_trace} depicts the trace plots and autocorrelation function (ACF) under a squared and unsquared distance-to-set term in the prior. The trace plots suggest better mixing and slightly less stickiness in the sampling trajectories. Moreover, there is a noticeable reduction in dependence between sample draws when using the squared distance-to-set priors based on the ACF plots. Overall, this is a relatively simple example---the dimension of the constraint set matches the dimensions of the ambient space within which it is embedded. In fact, both squared and unsquared priors perform well, and the samples from the latter look essentially the same as Figure \ref{fig:opt}, though we already see computational improvements by examining properties of the chain. These differences become more pronounced as we consider more challenging settings, such as a lower dimensional constraint set in the next example. 

\begin{figure}
    \centering
        \includegraphics[width=0.9\textwidth]{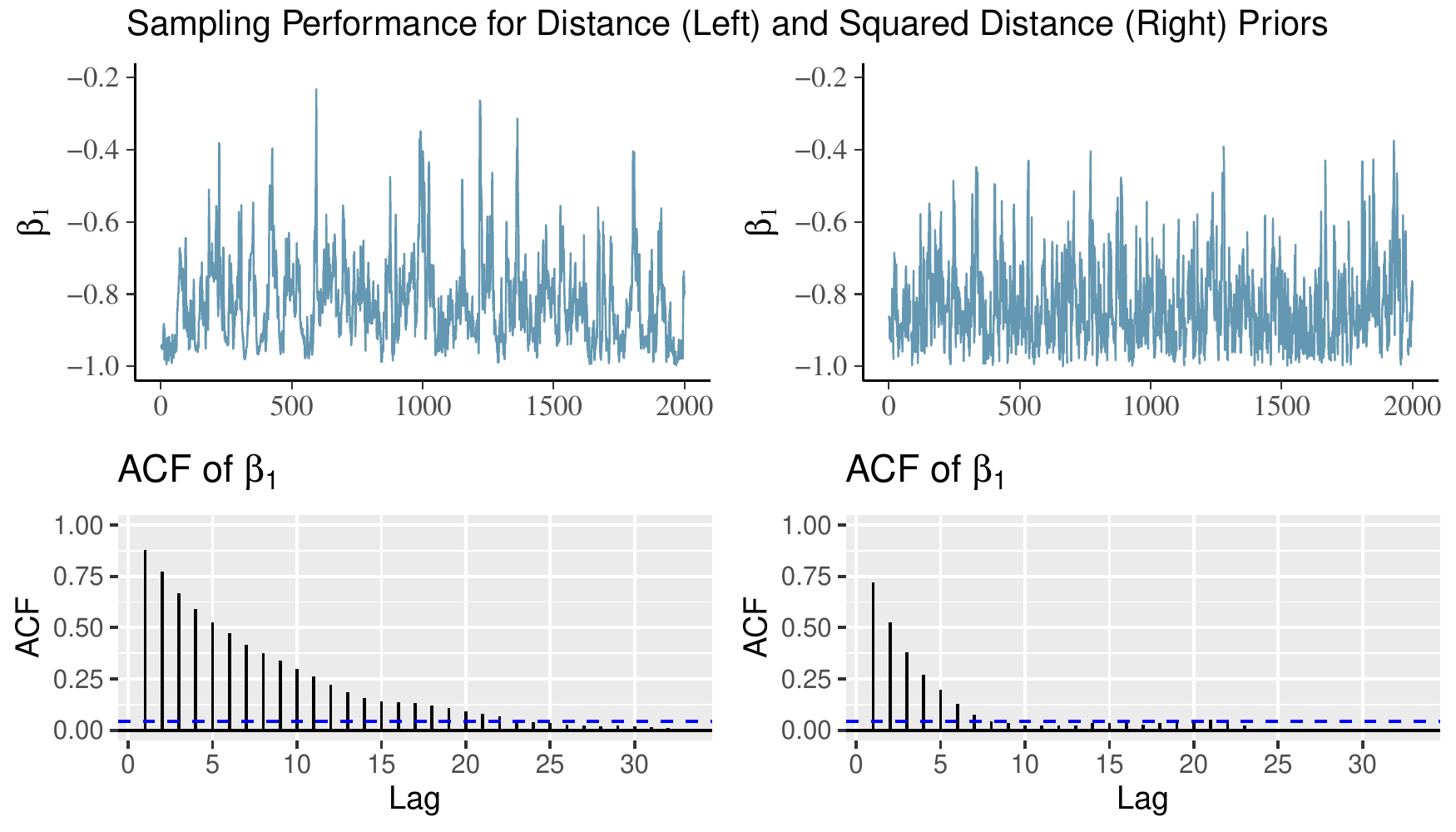}
    \caption{(Regression) Trace plots and ACF plots for both the unsquared (left) and squared (right) distance-to-set priors for $\beta_1$. Plots for $\beta_2$ look similar and are omitted.}
    \label{fig:acf_trace}
\end{figure}

\paragraph{Sampling along a Lower-Dimensional Surface}

The von Mises-Fisher $\textsf{vMF}(\alpha,\bF)$ distribution is supported on the sphere $S^p$ with $\alpha\geq 0$ and $\bF\in S^p$. When $\alpha=0$, this reduces to the uniform distribution on the sphere \citep{Fisher1953DispersionOA}. One can then envision the von Mises-Fisher distribution as being a spherical distribution concentrated around a unit vector. \citet{duan2020bayesian} observe that this distribution can be described as a multivariate normal with mean vector $\bF\in\bbR^{p+1}$ constrained to the unit sphere. They further consider a generalization in which the multivariate normal likelihood is replaced with a multivariate Student-$t$ distribution with $m$ degrees of freedom, mean vector $\bF\in\bbR^{p+1}$, and variance $\sigma^2\bI_{p+1}$. Using distance-to-set priors, we revisit this setting and relax the constraint so that the points have to lie close to the constraint surface, targeting sampling from the following distribution:

$$
\widetilde{\pi}(\btheta\mid\by) \propto \left(1 + \frac{\|\bF - \btheta\|_2^2}{m\sigma^2} \right)^{-\frac{m+p}{2}}\exp\left(-\frac{\rho}{2}\dist(\btheta,S^p)^2 \right)
$$

Observe that $S^2$ has a smaller dimension that the space within which it is embedded, namely $\bbR^3$. We demonstrate how one would use distance-to-set constraint relaxation in this setting, we specify the the projection $P_{S^p}$, which maps $0\neq \btheta\in\bbR^n$ to $\btheta/\|\btheta\|$. Thus, $\dist(\btheta,S^p) = \|\btheta - P_{S^p}(\btheta)\|$.
In \citet{duan2020bayesian}, the distance from the constraint is considered algebraically $\nu(\btheta) = |\btheta^\intercal\btheta - 1|$. In essence, the distance from $\btheta$ to $S^p$ is given by the distance in the level curve it lies on $\btheta^\intercal\btheta$ and the level curve defining $S^p$. As such, we refer to this as the level set relaxation prior in comparisons reported here.\\

We now compare these two Bayesian constraint relaxation approaches using \texttt{stan}. For sampling using the relaxation $\nu(\btheta)$, we use publicly accessible code obtainable in \citet{duan}. For our distance-to-set prior, we only need to update the constraint relaxation term. We summarize how these Bayesian constraint relaxation methods perform below for $p=2$ (the sphere that forms the boundary of the Euclidean $\ell_2$ unit ball in $\bbR^3$). Figure \ref{fig:vMF_plot} plots results after drawing 2000 samples points using the peer algorithms, thinned by a factor of 10 for visual clarity. The distance-to-set prior mimics the theoretical draws from the $\textsf{vMF}$ distribution, with some deviation due to a mild degree of constraint relaxation and slightly different tail behavior. In contrast, the level set relaxation prior leads to a chain that gets stuck during sampling, and the range of samples do not appear to be near the target constraint surface.

\begin{figure*}
    \centering
    \begin{subfigure}{0.32\textwidth}
         \includegraphics[width=\textwidth]{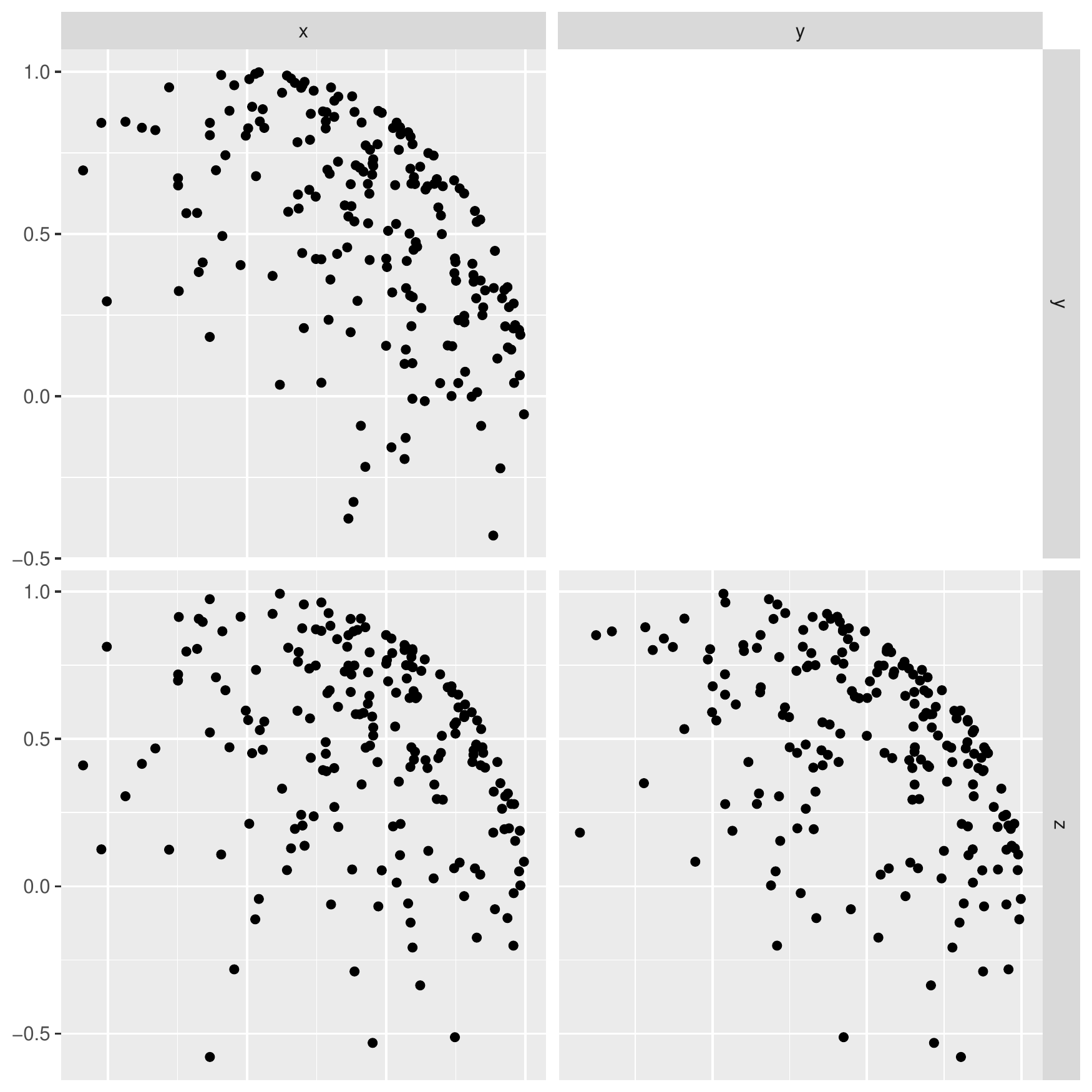}
         \caption{Theoretical Draws (vMF)}
     \end{subfigure}
     \hfill
     \begin{subfigure}{0.32\textwidth}
         \includegraphics[width=\textwidth]{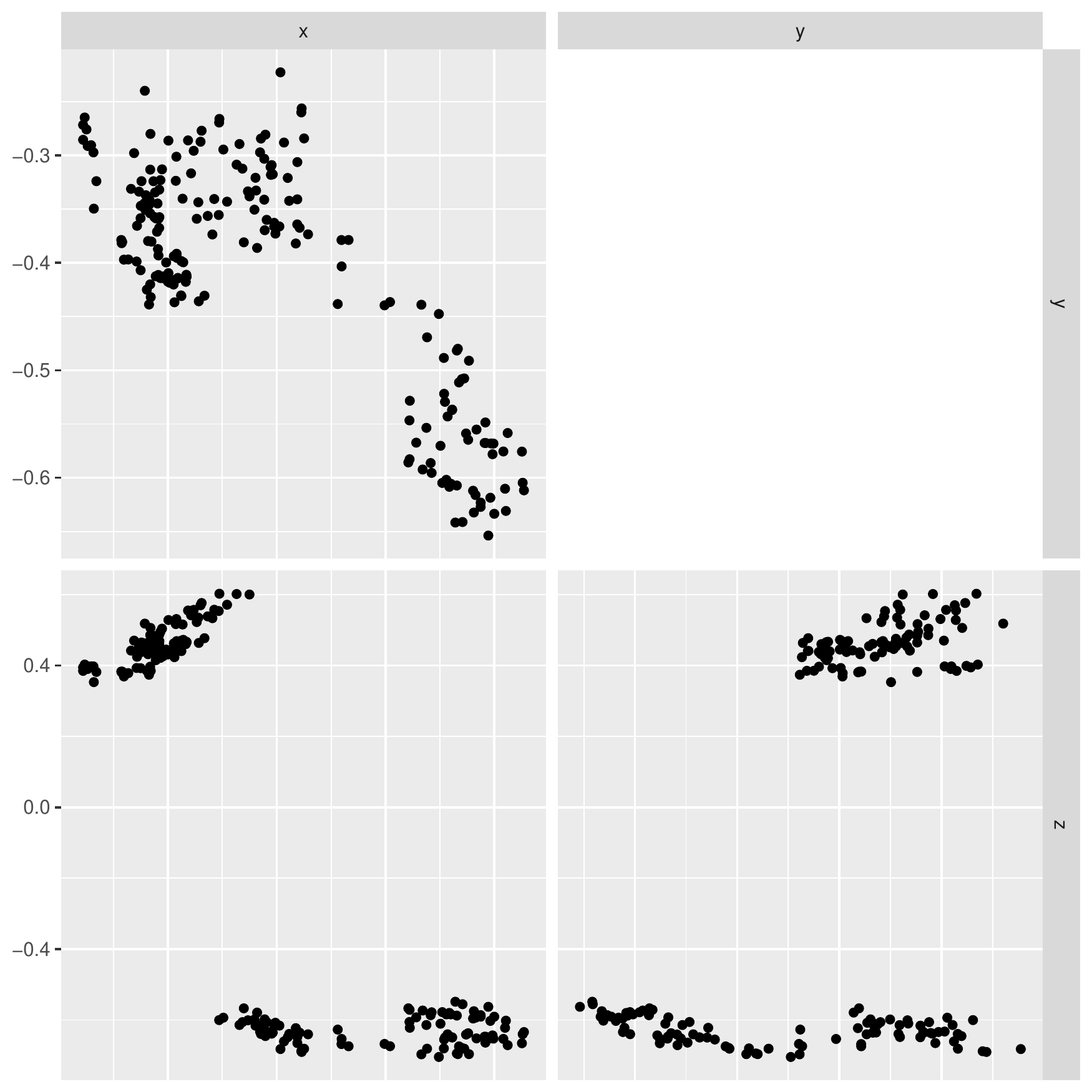}
         \caption{Level Set (RvMF)}
     \end{subfigure}
     \hfill
     \begin{subfigure}{0.32\textwidth}
         \includegraphics[width=\textwidth]{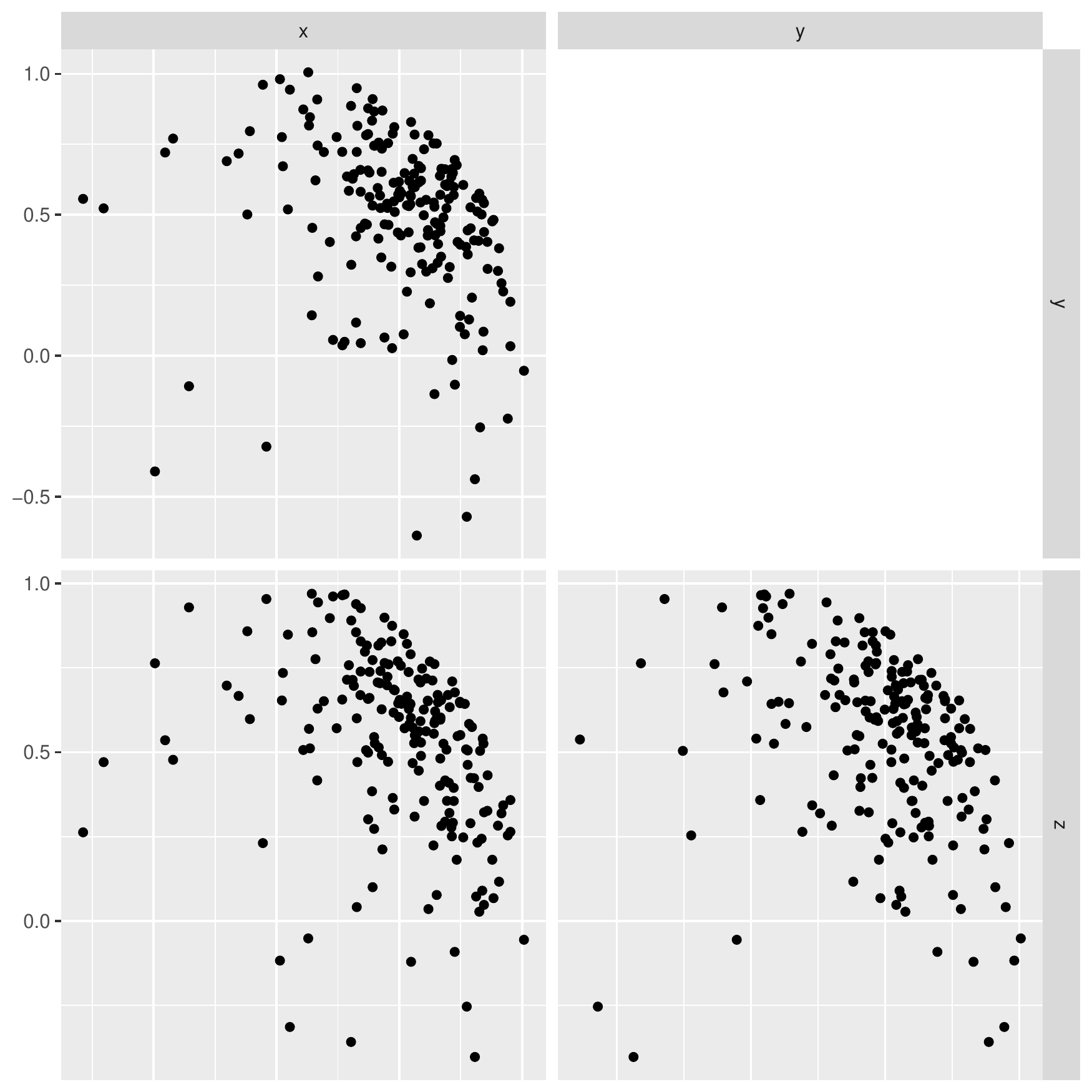}
         \caption{Distance-to-Set (RvMF)}
     \end{subfigure}     
     \caption{
     Exact draws using the using the \texttt{rvmf} function in the \texttt{rFast} package compared to samples using the method of \citet{duan2020bayesian} and our proposed method under $\rho = 10^{5}$ with $\bF = (1/\sqrt{3},1/\sqrt{3},1/\sqrt{3})$, $\sigma^2 = 0.1$, and $m=3$. The left plot is from the vMF distribution, while the middle and right plots are from the Robust vMF distribution.  
     } \vspace{-10pt}
    \label{fig:vMF_plot}
\end{figure*}

\begin{table}[ht]
\centering
\caption{Sampling Performance for Level Set Prior vs. Distance-to-Set Prior on Robust vMF Distribution}
\scalebox{0.95}{
\begin{tabular}{@{\extracolsep{4pt}}lc | rrrr  |rrrr@{}}
  \multicolumn{2}{c}{}& \multicolumn{4}{c}{Level Set Relaxation Prior} & \multicolumn{4}{c}{Distance-to-Set Prior} \\
  \cline{3-6}  \cline{7-10}
$\rho$ & Axis & Mean & 2.5\% & 97.5\% & ESS & Mean & 2.5\% & 97.5\% & ESS \\ 
  \midrule
1000 & x & 0.59 & 0.18 & 0.94 & 31.55 & 0.52 & -0.02 & 0.93 & 853.22 \\ 
   & y & 0.54 & 0.05 & 0.86 & 9.38 & 0.52 & -0.08 & 0.93 & 736.91 \\ 
   & z & 0.48 & 0.04 & 0.88 & 11.78 & 0.53 & -0.11 & 0.96 & 728.31 \\ 
  10000 & x & 0.26 & -0.10 & 0.57 & 1.35 & 0.51 & -0.13 & 0.92 & 750.94 \\ 
   & y & 0.41 & -0.13 & 0.78 & 1.05 & 0.51 & -0.14 & 0.93 & 650.81 \\ 
   & z & -0.75 & -1.00 & -0.47 & 1.02 & 0.51 & -0.09 & 0.93 & 622.50 \\ 
  1e+05 & x & 0.27 & 0.07 & 0.46 & 1.00 & 0.50 & -0.21 & 0.92 & 751.80 \\ 
   & y & 0.06 & -0.88 & 0.99 & 1.00 & 0.50 & -0.29 & 0.94 & 600.63 \\ 
   & z & -0.01 & -0.14 & 0.12 & 1.00 & 0.49 & -0.22 & 0.92 & 702.80 \\ 
  1e+06 & x & -0.38 & -0.46 & -0.30 & 1.00 & 0.52 & -0.02 & 0.92 & 779.55 \\ 
   & y & 0.51 & 0.06 & 0.95 & 1.00 & 0.51 & -0.15 & 0.92 & 559.85 \\ 
   & z & -0.40 & -0.89 & 0.08 & 1.00 & 0.53 & -0.06 & 0.93 & 542.38 \\ 
   \bottomrule
\end{tabular}
}
\label{tab:vMF_tbl1}
\end{table}

Figure \ref{fig:vMF_plot} shows a cluster of sample draws for the $d$-expansion prior, suggesting that sampler is sticky and does not explore the space well. On the other hand, our distance-to-set prior resembles draws from the theoretical distribution fairly well. Table \ref{tab:vMF_tbl1} further reinforces this point. As we decrease $\lambda$ (i.e., enforce the constraint more strictly), we see that the distribution concentrates away from the mean vector $(1/\sqrt{3}, 1/\sqrt{3}, 1/\sqrt{3})$, and the ESS (out of 1,000 post-warmup iterations per chain, 2 chains) is low. Our novel distance-to-set prior concentrates around the mean vector, and the ESS remains consistently high. Finally, the acceptance rate for our method ranges between 0.932---0.937, while it ranges between 0.766---0.815 for the level set relaxation prior. As it is desirable for acceptance rates to be close to 100\% since Metropolis steps in HMC are meant to correct only for numerical error, this makes a strong case for the sampling performance under our approach.

\paragraph{Real Data Case Study}
While the simulation studies highlight advantages of our approach, the final example considers a case study whose constraint is nontrivial to incorporate within prior methods. We apply our distance-to-set priors  to constraints imposed on contingency tables imbued with isotonic constraints. We follow the design introduced by \cite{agresti2002analysis} in which four treatment group doses were given (Placebo, Low dose, Medium dose, High dose) to patients with subarachnoid hemorrhage, and the outcomes were examined (Good recovery, Minor disability, Major disability, Vegetative state, and Death) to construct a dose-response curve. The data appears in \cite{agresti2002analysis}, summarized in the table below. The constraint on this table that is natural to assume is for the outcome to stochastically increase with respect to the treatment, which we formalize below.\\

A model for the order-based constraints on this contingency table is given by \cite{sen2018constrained}. Following this treatment, suppose we have $n$ observations exhaustively distributed over an $I\times J$ contingency table with entries $n_{ij}$ for $i\in[I]$ and $j\in[J]$, where $[m] := \{1,\ldots,m\}$. We let the rows represent the doses, and the columns represent the outcomes. Suppose further that the probability of each observation ending up in the $(i,j)$ cell is given by $\theta_{ij}$. Let $n_{[i]} := \sum_{i\in[I]} n_{ij}$, and similarly, $\btheta_{[i]} := (\theta_{i1},\ldots,\theta_{iJ})$:
\begin{center}
    \scalebox{0.9}{
    \begin{tabular}{c|ccccc}
        & Recovery & Vegetative State & Major Disability & Minor Disability & Death \\\hline
        Placebo & $\theta_{11}$ & $\theta_{12}$ & $\theta_{13}$ & $\theta_{14}$ & $\theta_{15}$ \\
        Low Dose & $\theta_{21}$ & $\theta_{22}$ & $\theta_{23}$ & $\theta_{24}$ & $\theta_{25}$ \\
        Medium Dose & $\theta_{31}$ & $\theta_{32}$ & $\theta_{33}$ & $\theta_{34}$ & $\theta_{35}$ \\
        High Dose & $\theta_{41}$ & $\theta_{42}$ & $\theta_{43}$ & $\theta_{44}$ & $\theta_{45}$ \\
    \end{tabular}}
\end{center}
Then we take the following model: for each $i\in[I]$ suppose 
$$
(n_{i1},\ldots,n_{iJ}) \overset{\indep}{\sim} \mathsf{Multi}(n_{[i]}, \btheta_{[i]}),\quad \btheta_{[i]}\overset{\indep}{\sim}\textsf{Dir}(\balpha).
$$
We impose the stochastic dominance constraint on the probabilities governing the contingency table as follows: for all $i\in[I]$, for all $j\in[J]$,
$$
\sum_{k=1}^j \theta_{i+1,k} \geq \sum_{k=1}^j \theta_{ik}.
$$
We may write the set of such probabilities obeying this stochastic dominance  as the following isotonic constraint:
$$
\Theta_{CT} := \left\{(\theta_{ij})_{i\in I, j\in J}\Bigg|\sum_{k=1}^j \theta_{i+1,k} \geq \sum_{k=1}^j \theta_{ik}  \text{ for } i\in[I], j\in[J]\right\}
$$

As with any application of distance-to-set priors, a crucial subroutine requires computing the projection onto $\bTheta_{CT}$.
It is not clear whether implementing the projection directly in a Stan file \citep{stan} is possible; instead, we implement our HMC-based sampler  directly in \texttt{R}. We use a quadratic programming algorithm \citep{goldfarb1982dual, goldfarb1983numerically}
available in the \texttt{quadprog} library \citep{turlach2013quadprog} to compute the projections that appear in the gradient of the log-posterior, and include  the complete implementation details to the Appendix. It is worth noting that despite four seemingly independent multinomial-Dirichlet models, the stochastic dominance constraints entangle the distributions. The resulting constrained problem is complex, and distinct from a standard setting with separate isotonic constraints \citep{chatterjee2015risk}, which can be handled using the simpler pooled adjacent violators algorithm (PAVA).\\

\begin{figure}
    \centering
        \includegraphics[width=0.75\textwidth]{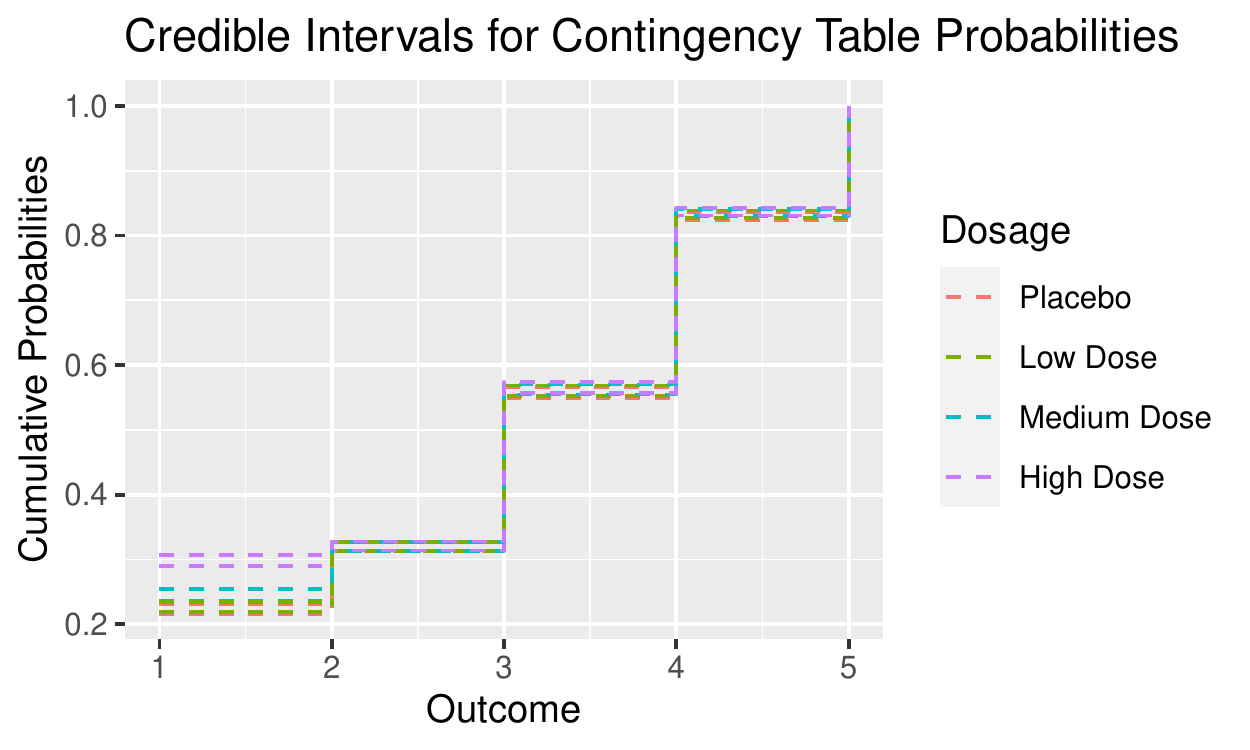}
    \caption{(Contingency Table) 95\% posterior credible intervals under distance-to-set priors with   $\rho=7.5\times 10^5$.}
    \label{fig:cred_int_plot}
\end{figure}

Figure \ref{fig:cred_int_plot} displays the 95\% credible intervals for the probabilities governing the contingency tables; detailed numerical values are tabulated in the Appendix. As we consider  a large value of $\rho=7.5\times 10^5$, it is not surprising that the isotonic constraints are well-respected at the quantiles. Despite the high degree of constraint enforcement, Figure \ref{fig:contingency_trace} shows that our approach to constraint relaxation maintains strong performance despite a na\"ive implementation.

\begin{figure}
    \centering
        \includegraphics[width=0.75\textwidth]{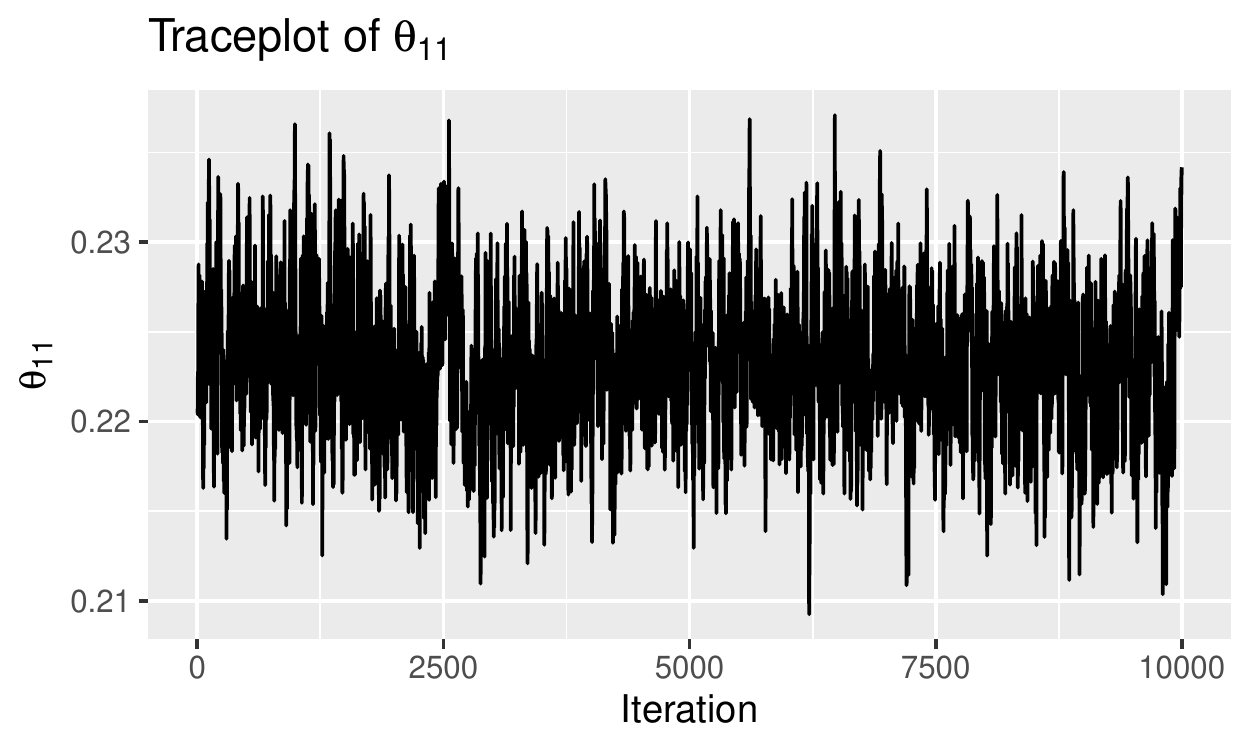}
    \caption{(Contingency Table) The trace plot of samples for the parameter $\theta_{1,1}$ in the contingency table.}
    \label{fig:contingency_trace}
\end{figure}

\section{Discussion}

In this work, we recast distance-to-set regularization within a Bayesian framework, and propose a flexible class of distance-to-set priors to allow uncertainty quantification under a well-defined posterior. Our class of priors is smooth, which is a crucial factor in the improved sampling performance under implementations such as HMC. Empirical results reflect this design, so that performance does not deteriorate even when the constraint set is a lower-dimensional manifold, the parameter $\rho$ is large, or $P_{\bTheta}(\btheta)$ cannot be expressed  analytically in closed form. When $\rho\to\infty$, distance-to-set priors agree with the sharply constrained Bayesian inference problem, and when $\rho<\infty$, constraint-relaxed posteriors are optimal approximations in terms of KL divergence, while respecting the constraint to a specified amount in expectation.\\

There is a number of open inferential extensions that remain interesting directions for future work under our distance-to-set regularization approach. For instance, loss functions do not always originate from likelihoods, so there is value in extending our framework to more general settings such as Gibbs posteriors \citep{bissiri2016general}. While our distance-to-set framework assumes an $\ell_2$ metric, considering other measures and divergences may be of particular interest for various data settings. Indeed, this is related to the underexplored connection between constraint relaxation and information geometry that we begin to examine via connections to exponential tilting. We invite readers to consider future investigation of these promising research directions.

\section*{Acknowledgments}
We are grateful for insightful discussions with Mike West and Emily Tallman on exponential tilting and KL divergence.

\bibliographystyle{plainnat}
\bibliography{bibliography}

\pagebreak

\section*{Appendix}

\subsection*{Numerical Data For Stochastic Ordering Case Study}

We provide numerical data supporting the credible intervals shown in Figure 5 of the paper.

\begin{table}[ht]
\centering
\begin{tabular}{r|rrrrr}
  \multicolumn{6}{c}{$2.5^\text{th}$ Percentile}\\
  \hline
 & $j=1$ & $j=2$ & $j=3$ & $j=4$ & $j=5$ \\ 
  \hline
$i=1$ & 0.2161 & 0.3126 & 0.5496 & 0.8229 & 1.0000 \\ 
  $i=2$ & 0.2194 & 0.3131 & 0.5527 & 0.8272 & 1.0000 \\ 
  $i=3$ & 0.2365 & 0.3134 & 0.5544 & 0.8300 & 1.0000 \\
  $i=4$ & 0.2900 & 0.3137 & 0.5566 & 0.8306 & 1.0000 \\ 
   \hline
\end{tabular}
\end{table}

\begin{table}[ht]
\centering
\begin{tabular}{r|rrrrr}
\multicolumn{6}{c}{$50^\text{th}$ Percentile}\\
  \hline
 & $j=1$ & $j=2$ & $j=3$ & $j=4$ & $j=5$ \\ 
  \hline
$i=1$ & 0.2234 & 0.3194 & 0.5581 & 0.8297 & 1.0000 \\ 
  $i=2$ & 0.2266 & 0.3199 & 0.5605 & 0.8329 & 1.0000 \\ 
  $i=3$ & 0.2456 & 0.3202 & 0.5621 & 0.8356 & 1.0000 \\ 
  $i=4$ & 0.2983 & 0.3206 & 0.5646 & 0.8360 & 1.0000 \\ 
   \hline
\end{tabular}
\end{table}

\begin{table}[ht]
\centering
\begin{tabular}{r|rrrrr}
\multicolumn{6}{c}{$97.5^\text{th}$ Percentile}\\
  \hline
 & $j=1$ & $j=2$ & $j=3$ & $j=4$ & $j=5$ \\ 
  \hline
$i=1$ & 0.2310 & 0.3265 & 0.5663 & 0.8358 & 1.0000 \\ 
  $i=2$ & 0.2346 & 0.3268 & 0.5683 & 0.8384 & 1.0000 \\
  $i=3$ & 0.2546 & 0.3271 & 0.5698 & 0.8412 & 1.0000 \\ 
  $i=4$ & 0.3065 & 0.3277 & 0.5734 & 0.8418 & 1.0000 \\ 
   \hline
\end{tabular}
\end{table}

\subsection*{Proofs}

\setcounter{prop}{0}
\begin{prop}
Under Assumptions 1 and 2, the constraint relaxed posterior $\widetilde{\pi}(\btheta\mid\by)$ is a proper density.
\end{prop}

\begin{proof}
Noting that $\exp\left(-\frac{\rho}{2}\dist(\btheta,\bTheta)^2\right)\leq 1$ for all $\btheta\in\bbR^d$, we immediately obtain from Assumption 2
$$
\int_{\bbR^d} L(\btheta\mid\by)\pi(\btheta)\exp\left(-\frac{\rho}{2}\dist(\btheta,\bTheta)^2\right)\,d\btheta \leq \int_{\bbR^d} L(\btheta\mid\by)\pi(\btheta)\,d\btheta < \infty.
$$
\end{proof}

\setcounter{thm}{0}
\begin{thm}
Suppose the unconstrained posterior $\pi(\btheta\mid \by)$ is strictly log-concave. Let $\{\widetilde{\pi}_{\rho_k}(\btheta\mid\by)\}_{k\in\bbN}$ be a sequence of constraint-relaxed posterior distributions where $\rho_k\uparrow\infty$ as $k\to\infty$. Further, define the following MAP estimators
\[ \widehat{\btheta}^M = \argmax_{\btheta} \overline{\pi}(\btheta\mid\by),
   \quad  \widehat{\btheta}_{\rho_k}^M = \argmax_{\btheta} \widetilde{\pi}_{\rho_k}(\btheta\mid\by).\]
Then the sequence $\widehat{\btheta}_{\rho_k}^M \to \widehat{\btheta}^M$ as $k\to\infty$.
\end{thm}

\begin{proof}
We closely follow the argument of \cite[Theorem~17.1]{wright1999numerical}.
Let $\btheta^*$ be a limit point of the sequence of MAP estimates $\left\{\widehat{\btheta}_{\rho_k}^M\right\}_{k\in\bbN}$ corresponding to the constraint-relaxed posterior distributions $\widetilde{\pi}_{\rho_k}(\btheta\mid\by)$, and let $\widehat{\btheta}^M$ be the MAP estimate corresponding to the constrained posterior distribution $\overline{\pi}(\btheta\mid\by)$. For the remainder of the proof, we drop the superscript $M$ to ease notation.

We start with the basic inequality $\log\widetilde{\pi}_{\rho_k}\left(\widehat{\btheta}\mid\by\right)\leq\log\widetilde{\pi}_{\rho_k}\left(\widehat{\btheta}_{\rho_k}\mid\by\right)$ and expand:
$$
\log\pi\left(\widehat{\btheta}\mid\by\right) - \frac{\rho}{2}\underbrace{\dist\left(\widehat{\btheta},\bTheta\right)^2}_{=0} +\log C_{\rho_k} \leq \log\pi\left(\widehat{\btheta}_{\rho_k}\mid\by\right) - \frac{\rho_k}{2}\dist\left(\widehat{\btheta}_{\rho_k},\bTheta)^2 \right) + \log C_{\rho_k},
$$
where $C_{\rho_k} = \left(\int_{\bbR^d} \widetilde{\pi}_{\rho_k}(\btheta\mid\by)\,d\btheta \right)^{-1}$ is the normalizing constant for $\widetilde{\pi}(\btheta\mid\by)$. Simplifying the above expression gives
\[
\log \pi\left(\widehat{\btheta}\mid\by \right) \leq \log\pi\left(\widehat{\btheta}_{\rho_k}\mid\by\right) - \frac{\rho_k}{2}\dist\left(\widehat{\btheta}_{\rho_k},\bTheta\right)^2.\tag{*}
\]
A bit of algebra shows that
$$
0 \leq \dist\left(\widehat{\btheta}_{\rho_k}\mid\by\right)^2 \leq \frac{2}{\rho_k}\left(\log\pi\left(\widehat{\btheta}_{\rho_k}\mid\by\right) - \log\pi\left(\widehat{\btheta}\mid\by \right) \right)
$$
Since $\rho_k\uparrow\infty$ as $k\to\infty$ and $\log\pi(\cdot\mid\by)$ is continuous, we have that
$$
\frac{2}{\rho_k}\left(\log\pi\left(\widehat{\btheta}_{\rho_k}\mid\by\right) - \log\pi\left(\widehat{\btheta}\mid\by \right) \right)\overset{k\to\infty}{\longrightarrow} 0,
$$
which implies that
$$
\dist(\btheta^*,\bTheta) = \lim_{k\to\infty}\dist\left(\widehat{\btheta}_{\rho_k},\bTheta\right) = 0.
$$
The equality in the above expression follows from the continuity of $\dist(\cdot,\bTheta)$ \citep[Proposition~2.7.1(a)]{lange2016mm}. Thus, $\btheta^*\in\bTheta$.

Next, since $\rho>0$ and $\dist(\cdot,\bTheta)\geq 0$, we have for all $k\in\bbN$,
$$
\log\pi\left(\btheta^*\mid\by \right) \geq \log\pi\left(\btheta^*\mid\by \right) - \frac{\rho_k}{2}\dist\left(\widehat{\btheta}_{\rho_k},\bTheta \right)^2,
$$
so we must have
\[
\log\pi\left(\btheta^*\mid\by \right) \geq \lim_{k\to\infty}\left(\log\pi\left(\btheta^*\mid\by \right) - \frac{\rho_k}{2}\dist\left(\widehat{\btheta}_{\rho_k},\bTheta \right)^2 \right) = \log\pi\left(\btheta^*\mid\by \right) - \lim_{k\to\infty}\frac{\rho_k}{2}\dist\left(\widehat{\btheta}_{\rho_k},\bTheta \right)^2\tag{**}
\]
Taking $k\to\infty$ on both sides of (*), we obtain
\[
\log\pi\left(\widehat{\btheta}\mid\by \right) \leq \lim_{k\to\infty}\left(\log\pi\left(\btheta^*\mid\by \right) - \frac{\rho_k}{2}\dist\left(\widehat{\btheta}_{\rho_k},\bTheta \right)^2 \right) = \log\pi\left(\btheta^*\mid\by \right) - \lim_{k\to\infty}\frac{\rho_k}{2}\dist\left(\widehat{\btheta}_{\rho_k},\bTheta \right)^2\tag{***}
\]
Combining (**) and (***), we obtain the inequality:
$$
\log\pi\left(\btheta^*\mid\by\right) \geq \log\pi\left(\widehat{\btheta}\mid\by\right).
$$
By the strict log-concavity of $\pi(\btheta\mid\by)$ and the convexity of $\bTheta$ (Assumption 3), we have that $\widehat{\btheta}$ is the unique global maximum of $\overline{\pi}(\btheta\mid\by)$. Therefore, $\btheta^* = \widehat{\btheta}$, and the conclusion follows.
\end{proof}

\begin{thm}
Let $\overline{\Pi}$ be the constrained posterior distribution with density $\overline{\pi}(\btheta\mid\by)$, and let $\{\widetilde{\Pi}_{\rho_k}\}_{k\in\bbN}$ be a sequence of constraint-relaxed posterior distributions with densities $\{\widetilde{\pi}_{\rho_k}(\btheta\mid\by)\}_{k\in\bbN}$, respectively, where $\rho_k\uparrow\infty$ as $k\to\infty$. Then  $\|\widetilde{\Pi}_{\rho_k} - \overline{\Pi}\|_{\mathsf{TV}} \to 0$ as $k\to \infty$. It follows that $\widetilde{\Pi}_{\rho_k}\overset{D}{\rightarrow}\overline{\Pi}$ as $k\to 0$.
\end{thm}

\begin{proof}
By Assumption 1, the constrained distribution $\overline{\Pi}$ and the constrained relaxed posterior distributions $\widetilde{\Pi}_{\rho_k}$, $k\in\bbN$, are absolutely continuous, so they have densities $\overline{\pi}(\btheta\mid\by)$ and $\widetilde{\pi}_{\rho_k}(\btheta\mid\by)$, respectively. Write these as
\begin{align*}
    \overline{\pi}(\btheta\mid\by) & = CL(\btheta\mid\by)\pi(\btheta)\1_{\btheta\in\bTheta} \\
    \widetilde{\pi}_{\rho_k}(\btheta\mid\by) & = C_{\rho_k}L(\btheta\mid\by)\pi(\btheta)\exp\left(-\frac{\rho_k}{2}\dist(\btheta,\bTheta)^2 \right),
\end{align*}
where $C$ and $C_{\rho_k}$ are normalizing constants. Since $\1_{\bTheta}\leq \exp\left(-\frac{\rho}{2}\dist(\btheta,\bTheta)^2\right)$ for all $\btheta\in\bbR^d$ (with equality on $\bTheta$), we have
\begin{align*}
    C_{\rho_k}^{-1} & = \int_{\bbR^d} L(\btheta\mid\by)\pi(\btheta)\exp\left(-\frac{\rho_k}{2}\dist(\btheta,\bTheta)^2\right)d\,\btheta \\
    & \geq \int_{\bTheta} L(\btheta\mid\by)\pi(\btheta)\exp\left(-\frac{\rho_k}{2}\dist(\btheta,\bTheta)^2\right)d\,\btheta \\
    & = \int_{\bTheta} L(\btheta\mid\by)\pi(\btheta)\1_{\btheta\in\bTheta}d\,\btheta \\
    & = C^{-1}
\end{align*}
Thus, $C_{\rho_k}\leq C$ for all $k\in\bbN$. By a similar calculation, we have
\begin{align*}
    C_{\rho_{k+1}}^{-1} & = \int_{\bbR^d} L(\btheta\mid\by)\pi(\btheta)\exp\left(-\frac{\rho_{k+1}}{2}\dist(\btheta,\bTheta)^2\right)d\,\btheta \\
    & \geq \int_{\bbR^d} L(\btheta\mid\by)\pi(\btheta)\exp\left(-\frac{\rho_{k}}{2}\dist(\btheta,\bTheta)^2\right)d\,\btheta \\
    & = C_{\rho_k}^{-1}.
\end{align*}
We then have $C_{\rho_{k+1}}\leq C_{\rho_k}\leq C$ for all $k\in\bbN$. Therefore, $C_{\rho_k}\uparrow C$. Partition $\bbR^d = \bTheta \cup \bTheta^C$, and observe that for $\btheta\in\bTheta^C$,
$$
\widetilde{\pi}_{\rho_k}(\btheta\mid\by) \geq 0 = \overline{\pi}(\btheta\mid\by).
$$
For $\btheta\in\bTheta$, for all $k\in\bbN$,
$$
\overline{\pi}(\btheta\mid\by) = CL(\btheta\mid\by)\pi(\btheta) \geq C_{\rho_k}L(\btheta\mid\by)\pi(\btheta) = \widetilde{\pi}_{\rho_k}(\btheta\mid\by)
$$
Thus, we have
$$
\bTheta^C = \{\btheta: \widetilde{\pi}_{\rho_k}(\btheta\mid\by) \geq \pi(\btheta\mid\by)\},\qquad \bTheta = \{\btheta: \widetilde{\pi}_{\rho_k}(\btheta\mid\by) \leq \pi(\btheta\mid\by)\},\quad k\in\bbN.
$$
Using this fact about the partition, we examine the TV distance:
\begin{align*}
    \|\widetilde{\Pi}_{\rho_k} - \overline{\Pi}\|_{\mathsf{TV}} & = \frac{1}{2}\int_{\bbR^d} |\widetilde{\pi}_{\rho_k}(\btheta\mid\by) - \overline{\pi}(\btheta\mid\by)|\,d\btheta \\
    & = \frac{1}{2}\int_{\bTheta^C} \widetilde{\pi}_{\rho_k}(\btheta\mid\by) - \overline{\pi}(\btheta\mid\by)\,d\btheta + \frac{1}{2}\int_{\bTheta} \overline{\pi}(\btheta\mid\by) - \widetilde{\pi}_{\rho_k}(\btheta\mid\by)\,d\btheta \\
    & = \underbrace{\frac{1}{2}\int_{\bTheta^C} \widetilde{\pi}_{\rho_k}(\btheta\mid\by)\,d\btheta}_{(a)} + \underbrace{\frac{1}{2}(C - C_{\rho_k})\int_{\bTheta} L(\btheta\mid\by)\pi(\btheta)\,d\btheta}_{(b)}
\end{align*}

For the first integral (a), observe that
$$
\lim_{k\to\infty}\exp\left(-\frac{\rho_k}{2}\dist(\btheta,\bTheta)^2\right) = \1_{\btheta\in\bTheta}
$$
in a pointwise manner for all $\btheta\in\bbR^d$, so for all $\btheta\in\bbR^d$, $\widetilde{\pi}_{\rho_k}(\btheta\mid\by)\downarrow 0$ as $k\to\infty$. Note that
$$
\widetilde{\pi}_{\rho_k}(\btheta\mid\by)  = C_{\rho_k}L(\btheta\mid\by)\pi(\btheta)\exp\left(-\frac{\rho_k}{2}\dist(\btheta,\bTheta)^2\right) \leq CL(\btheta\mid\by)\pi(\btheta) \in L^1
$$
by Assumption 2. Therefore, by the Monotone Convergence Theorem, we have
$$
\lim_{k\to\infty}\frac{1}{2}\int_{\bTheta^C} \widetilde{\pi}_{\rho_k}(\btheta\mid\by)\,d\btheta = \frac{1}{2}\int_{\bTheta^C} \lim_{k\to\infty}\widetilde{\pi}_{\rho_k}(\btheta\mid\by)\,d\btheta = 0.
$$

For the second integral (b), we have:
\begin{align*}
    \lim_{k\to\infty} \frac{1}{2}(C - C_{\rho_k})\int_{\bTheta} L(\btheta\mid\by)\pi(\btheta)\,d\btheta = 0
\end{align*}

Putting this all together, we obtain
$$
\|\widetilde{\Pi}_{\rho_k} - \overline{\Pi}\|_{\mathsf{TV}} \to 0\quad\text{as}\quad k\to\infty.
$$

This shows that $\widetilde{\Pi}_{\rho_k}$ converges to $\Pi$ in total variation distance. An equivalent way of writing this is to say that
$$
\lim_{k\to\infty}\sup_{\mathcal{S}\in\mathcal{B}(\bbR^d)} |\widetilde{\Pi}_{\rho_k}(\mathcal{S}) - \overline{\Pi}(\mathcal{S})| = 0,
$$
where we take the supremum over all Borel sets $\mathcal{B}(\bbR^d)$. Let $\mathcal{A}\subset\mathcal{B}(\bbR^n)$ be the collection of $\overline{\Pi}$-continuity sets (i.e., for all $A\in\mathcal{A}$, $\overline{\Pi}(\partial A)=0$). Clearly,
$$
\lim_{k\to\infty}\sup_{A\in\mathcal{A}} |\widetilde{\Pi}_{\rho_k}(A) - \overline{\Pi}(A)| \leq \lim_{k\to\infty}\sup_{A\in\mathcal{B}(\bbR^d)} |\widetilde{\Pi}_{\rho_k}(A) - \overline{\Pi}(A)| = 0.
$$
Then by the Portmanteau Theorem \citep[Theorem~2.1]{billingsley2013convergence}, we conclude that $\widetilde{\Pi}_{\rho_k} \overset{D}{\rightarrow} \overline{\Pi}$.

\end{proof}

\begin{thm}
Suppose that $\bbE_{\btheta\sim\pi(\btheta\mid\by)}[\dist(\btheta,\bTheta)^2/2] > D$. Then the constraint-relaxed posterior distribution $\widetilde{\pi}(\btheta\mid\by)$  is the solution to the moment-constrained information projection problem:
$$
p^*(\btheta) \propto \pi(\btheta\mid\by)\exp\left(-\frac{\lambda}{2}\dist(\btheta,\bTheta)^2 \right),
$$
where $\lambda>0$ is a Lagrange multiplier that satisfies the moment constraint under $p^*(\btheta)$.
\end{thm}

\begin{proof}
The general form of the moment-constrained information projection problem can be stated as follows:
\begin{align*}
    \min_p\; & \text{KL}(p\,\|\, q) := \int p\log\left(\frac{p}{q} \right) \\
    \st & \bbE_p[g(X)] = g_0,
\end{align*}
where $p$ and $q$ are density functions, $X\in\bbR^m$ is a random vector, $g:\bbR^m\to\bbR^d$ is a function of $X$, and $g_0\in\bbR^d$ is some given vector. The optimal solution is given by \cite{west2020perspectives} to be
$$
p^*(x)\propto q(x)\exp\left(\lambda^\intercal g(x) \right).
$$
See also \cite{csiszar1975divergence} and \cite{ robertson2005forecasting}. The optimal form of the solution, which is known as exponential tilting, and one can solve for the vector $\lambda$ using the moment constraint:
$$
\bbE_{p^*}[g(X)] = \int_{\bbR^m} g(x)q(x)\exp\left(\lambda^\intercal g(x) \right)\,dx = 0.
$$
Specializing to our case, we take $x=\btheta$, $q(\btheta) = \pi(\btheta\mid\by)$, $g(\btheta) = \frac{1}{2}\dist(\btheta,\bTheta)^2$, and $g_0 = D$. This gives
$$
p^*(\btheta)\propto \pi(\btheta\mid\by)\exp\left(\frac{\lambda}{2}\dist(\btheta,\bTheta)^2\right),
$$
where $\lambda\in\bbR$. Moreover, note that by \cite[Equation~5]{tallman2022entropic} and the remark immediately proceeding it, we have the relationship
$$
\frac{\partial}{\partial\lambda}\bbE_{p^*}\left[\frac{1}{2}\dist(\btheta,\bTheta)^2\right] > 0.
$$
Moreover, observe that $\lambda = 0$ if and only if $p^*(\btheta) = \pi(\btheta\mid\by)$. Thus, if we assume $\bbE_{\pi(\btheta\mid\by)}[\dist(\btheta,\bTheta)^2/2]>D$, then $\lambda<0$. Reparameterizing, we have
$$
p^*(\btheta)\propto \pi(\btheta\mid\by)\exp\left(-\frac{\lambda}{2}\dist(\btheta,\bTheta)^2\right),
$$
where $\lambda>0$, and this gives the desired result.
\end{proof}

\begin{prop}
The log constraint-relaxed posterior $\log \widetilde{\pi}(\btheta\mid \by)$ is continuously differentiable as long as the  log-posterior $\log\pi(\btheta\mid\by)$ is continuously differentiable in $\btheta$.
\end{prop}

\begin{proof}
By Assumption 3, $\bTheta$ is convex, so $P_{\bTheta}(\btheta)$ is single-valued. We then have $\nabla_{\btheta}\left[\frac{1}{2}\dist(\btheta,\bTheta)^2 \right] = \btheta - P_{\bTheta}(\btheta)$ \citep{lange2016mm} for any $\btheta$.\\

It remains to show that the gradient is continuous. It suffices to show that the projection operator $P_{\bTheta}(\btheta)$ is continuous. One way to see this is to observe that $P_{\btheta}(\btheta)$ is firmly nonexpansive (i.e., 1-Lipschitz) when $\bTheta$ is closed and convex, and we immediately draw the conclusion.\\

However, we record here a slightly more general result that is useful for more general settings. For this we define the term Chebyshev set \citep{wulbert1968continuity}; a set $C$ is Chebyshev if for all $x\in C$, $P_C(x)$ is a singleton. Theorem 3 of \cite{wulbert1968continuity} provides one characterization of continuity which occurs for Chebyshev sets: if $C\subset X$ is a locally compact, Chebyshev set in a Banach space, then $P_C$ is continuous if and only if $C$ is convex.

$\bTheta\subset\bbR^n$ is convex, and so it is Chebyshev, and $\bbR^n$ is a Banach space. Moreover, it is easy to check that $\bbR^n$ is locally compact and Hausdorff. Since $\bTheta$ is closed, Corollary 29.3 of \cite{munkres2000topology} gives us that $\bTheta$ is locally compact. Thus, by the above theorem, we can conclude that $P_{\bTheta}$ is continuous.\\

Thus, for any point $\btheta_0\subset\bTheta$, 
$$
\lim_{\btheta\to\btheta_0} \nabla_{\btheta}\left[\frac{1}{2}\dist(\btheta,\bTheta)^2 \right] = \lim_{\btheta\to\btheta_0}\; \left[\btheta - P_{\bTheta}(\btheta)\right] = \btheta_0 - P_{\bTheta}(\btheta_0)
$$
We also have that $\dist(\btheta,\bTheta)$ is identically $0$ on $\bTheta$, so it's gradient there will be $0$. In particular, if $\btheta_0\in\partial\bTheta\subset\bTheta$, then 
$$
\lim_{\btheta\to\btheta_0} \nabla_{\btheta}\left[\frac{1}{2}\dist(\btheta,\bTheta)^2 \right] = \btheta_0 - P_{\bTheta}(\btheta_0) = \btheta_0 - \btheta_0 = 0 = \left.\nabla_{\btheta}\left[\frac{1}{2}\dist(\btheta,\bTheta)^2 \right] \right|_{\btheta = \btheta_0\in\partial\bTheta}
$$

\end{proof}

\subsection*{HMC Implementation, Stochastic Ordering Case Study}

We outline the details of the HMC sampler used in the contingency table application. Suppose the prior distribution is
$$
\theta_{[i]} \overset{\indep}{\sim} \mathsf{Dir}(\alpha_{[i]}),\quad i\in[I],
$$
where the entries of $\alpha_{[i]}$ are positive. By a small abuse of notation, let us also denote matrix of $\theta_{ij}$ by $\theta$. Also, let $\Delta$ denote the $(J+1)$-simplex. Then the posterior distribution with prior is given by

\begin{align*}
\pi(\theta\mid n_{[1]},\ldots,n_{[I]}) & \propto \left(\prod_{i=1}^I\prod_{j=1}^J \theta_{ij}^{n_{ij}}\right)\left(\prod_{i=1}^I\prod_{j=1}^J \theta_{ij}^{\alpha_{ij}-1}\right)\exp\left(-\frac{\rho}{2}\dist(\theta, \Theta_{CT}) \right)\1_{\bigcap_{i=1}^I\{\theta_{[i]}\in\Delta \} } \\
& \propto \left(\prod_{i=1}^I\prod_{j=1}^J\theta_{ij}^{n_{ij}+\alpha_{ij}-1} \right)\exp\left(-\frac{\rho}{2}\|\theta - P_{\Theta_{CT}}(\theta)\|_F^2\right) \1_{\bigcap_{i=1}^I\{\theta_{[i]}\in\Delta \} }
\end{align*}
As mentioned above, given some $\theta$, $P_{\Theta_{CT}}(\theta)$ can be computed using the \texttt{quadprog} library in \texttt{R}. For the HMC algorithm, the (negative) potential energy is defined by the log-posterior, which is given by
$$
\log \pi(\theta\mid n_{[1]},\ldots,n_{[I]}) = C + \sum_{i=1}^I\sum_{j=1}^J (n_{ij} + \alpha_{ij} - 1)\log\theta_{ij} - \frac{\rho}{2}\|\theta - P_{\Theta_{CT}}(\theta)\|_F^2 + \log\left(\1_{\bigcap_{i=1}^I\{\theta_{[i]}\in\Delta \}}\right),
$$
where $C$ is a constant that does not depend on $\theta$. Additionally, we need to compute the gradient of the (negative) log-potential with respect to $\theta$. To resolve the fact that the (negative) log-potential, we set $\theta_{iJ} = 1 - \sum_{j=1}^{J-1} \theta_{ij}$ for $i\in[I]$. Let us denote the matrix $\theta$ with the $J^\text{th}$ column removed by $\widetilde{\theta}$ and the corresponding posterior by $\widetilde{\pi}$. Denote the space of $I\times(J-1)$ matrices $\widetilde{\theta}$ that satisfy the desired isotonic constraint by
$$
\widetilde{\Theta}_{CT} := \left\{(\theta_{ij})_{i\in [I], j\in [J-1]}\Bigg|\sum_{k=1}^j \theta_{i+1,k} \geq \sum_{k=1}^j \theta_{ik}  \text{ for } i\in[I], j\in[J-1]\right\}
$$
Note that the entries of elements of $\widetilde{\Theta}_{CT}$ are non-negative. However, the rows do not sum to less than 1. Then the reduced (negative) log-potential becomes:
\begin{align*}
\log \widetilde{\pi}(\widetilde{\theta}\mid n_{[1]},\ldots,n_{[I]}) & = C + \underbrace{\sum_{i=1}^I\left((n_{iJ} + \alpha_{iJ} - 1)\log\left(1 - \sum_{j=1}^{J-1} \theta_{ij} \right) + \sum_{j=1}^{J-1} (n_{ij} + \alpha_{ij}-1)\log\theta_{ij} \right)}_{=:f(\widetilde{\theta})} \\
& \qquad -\frac{\rho}{2}\|\widetilde{\theta} - P_{\widetilde{\Theta}_{CT}}(\widetilde{\theta})\|_F^2
\end{align*}
We compute the gradient of the first term component-wise and obtain
$$
\frac{\partial f(\widetilde{\theta})}{\partial\theta_{ij}} = \frac{n_{ij} + \alpha_{ij} - 1}{\theta_{ij}} - \frac{n_{iJ} + \alpha_{iJ} - 1}{1 - \sum_{j=1}^{J-1}\theta_{ij}}
$$
To compute the gradient of the second (penalty) term, we note that the projection may be trickier because of the above substitution for $\theta_{iJ}$. However, we show this is not the case. Define $\widehat{\Theta}_{CT}$ to be the canonical embedding of $\widetilde{\Theta}_{CT}$ into the space of $I\times J$ matrices where the isotonic constraint only holds for the first $J-1$ columns in each row; more formally, $\widetilde{\Theta}_{CT}\hookrightarrow\widehat{\Theta}_{CT}$, and
$$
\widehat{\Theta}_{CT} := \left\{(\theta_{ij})_{i\in [I], j\in [J]}\Bigg|\sum_{k=1}^j \theta_{i+1,k} \geq \sum_{k=1}^j \theta_{ik}  \text{ for } i\in[I], j\in[J-1]\right\}
$$
Note that since each row of $\theta$ is a probability distribution, we have $\sum_{k=1}^J \theta_{ik} = 1$ holds for all $i\in[I]$. Thus, the following condition holds trivially for all $i\in[I]$
$$
\sum_{k=1}^J \theta_{i+1,k} \geq \sum_{k=1}^J \theta_{ik}.
$$
Thus, we immediately obtain $\widehat{\Theta}_{CT} = \Theta_{CT}$. As a corollary, projecting onto $\widehat{\Theta}_{CT}$ is equivalent to simply projecting onto the original parameter space $\Theta_{CT}$ because the last column of $\theta$, which we removed with our substitution, does not have an impact on the projection. This means we can project onto $\widetilde{\Theta}_{CT}$ by projecting onto $\Theta_{CT}$ and ignoring the last column of elements. Finally, we ensure that each row of $\theta$ sums to 1 manually once we have determined a proposal for the first $J-1$ columns as described above.The gradient of the second term (the penalty term) is obtained then as:
$$
\nabla_\theta\left(\frac{\rho}{2}\dist(\widetilde{\theta},\widetilde{\Theta}_{CT} \right) = \nabla_\theta\left(\frac{\rho}{2}\dist(\theta,\Theta_{CT})^2 \right) = \rho(\theta - P_{\Theta_{CT}}(\theta)),
$$
where we ignore the $J^\text{th}$ column of $\theta$. We can use the potential function and the gradient to obtain proposals $\theta$; we ensure that each row of $\theta$ sums to 1 manually once we have determined a proposal for the first $J-1$ columns using the procedure described above.\\

Next, it is worth noting that the first term of the gradient has properties akin to the log-barrier function technique \citep{boyd2004convex}. In particular, note that
$$
\lim_{\theta_{ij}\to 0^+} \frac{\partial f(\widetilde{\theta})}{\partial\theta_{ij}} = +\infty\qquad\text{ and }\qquad \lim_{\theta_{ij}\to 1-} \frac{\partial f(\widetilde{\theta})}{\partial\theta_{ij}} = -\infty.
$$
Thus, our sampler will deviate away from the boundary of the simplex that supports the Dirichlet prior. However, due to the fact that in practice we use a discretized leapfrog integrator to solve the Hamilton PDEs along with the fact that the gradients can be wildly different for each $\theta_{ij}$ but we have a constant scalar $\epsilon$ step-size, it may be that the behavior of $f(\widetilde{\theta})$ may overshoot the simplex when updating the position and momentum in an HMC proposal. As a result, we address this using a rather simple approach by rejecting any proposals that fall outside of the simplex. We also terminate any position-momentum updates that by chance fall outside of the simplex since the log-posterior is not defined outside of the simplex. While these solutions may offer some recourse with implementing this sampler, more sophisticated techniques, such as Spherical HMC \cite{lan2014spherical} or through spherical transformation of the simplex \citep{betancourt2012cruising}, may result in an improved implementation. However, we believe our main purpose is to demonstrate how one may incorporate a distance-to-set prior rather than construct the best-performing sampler for this particular problem.

\end{document}